\documentclass[11pt,a4paper]{article}
\usepackage{amsthm}
\usepackage[margin=1in]{geometry}
\usepackage{amsfonts}
\usepackage{algorithm}
\usepackage{caption}
\usepackage{algpseudocode}
\usepackage{physics}
\usepackage{float}
\usepackage{hyperref}
\usepackage{cleveref}
\usepackage{color}
\DeclareCaptionFormat{myformat}{#3}
\newtheorem{lemma}{Lemma}
\newtheorem*{lemm*}{Lemma}

\newtheorem{theorem}{Theorem}

\providecommand{\keywords}[1]
{
  \small	
  \textbf{\textit{Keywords---}} #1
}

\DeclareMathOperator{\dom}{dom}
\newcommand{\PMAJ}{\operatorname{PMAJ}}
\renewcommand{\Re}{\mathbb{R}}
\newcommand{\degeps}{\deg_\epsilon}
\newcommand{\zoo}{\{0,1\}}
\newcommand{\zoon}{\zoo^n}

\newcommand{\moo}{\{-1,+1\}}

\newcommand{\VV}{\mathcal{V}}
\newcommand{\oneton}{\{1,2,\dots,n\}}

\newcommand{\CT}{$2\textsc{p}\mathrm{CT}_{n,t,\epsilon}$} 
\newcommand{\CTc}{$2\textsc{p}\mathrm{CT}_{n,t,1/2}$} 
\newcommand{\CTs}{$2\textsc{p}\mathrm{CT}_{n,t,\epsilon}^{\gamma}$} 
\newcommand{\CTsc}{$2\textsc{p}\mathrm{CT}_{n,t,1/2}^{\gamma}$}

\begin{document}

\title{Quantum Communication Complexity of Distribution Testing}
\author{
  Aleksandrs Belovs$^{1}$,
  Arturo Castellanos$^{2}$,
  Fran{\c c}ois Le Gall$^{3}$,\\
  Guillaume Malod$^{4}$,
  Alexander A. Sherstov$^{5}$\\
  \small $^{1}$University of Latvia \\
  \small $^{2}$Kyoto University \\
  \small $^{3}$Nagoya University \\
  \small $^{4}$Universit{\' e} de Paris \\
    \small $^{5}$University of California, Los Angeles
  }



\maketitle

\begin{abstract}
The classical communication complexity of testing closeness of discrete distributions has recently been studied by Andoni, Malkin and Nosatzki (ICALP'19). In this problem, two players each receive $t$ samples from one distribution over $[n]$, and the goal is to decide whether their two distributions are equal, or are $\epsilon$-far apart in the $l_1$-distance. In the present paper we show that the quantum communication complexity of this problem is $\tilde{O}(n/(t\epsilon^2))$ qubits when the distributions have low $l_2$-norm, which gives a quadratic improvement over the classical communication complexity obtained by Andoni, Malkin and Nosatzki. We also obtain a matching lower bound by using the pattern matrix method. Let us stress that the samples received by each of the parties are classical, and it is only communication between them that is quantum. Our results thus give one setting where quantum protocols overcome classical protocols for a testing problem with purely classical samples.
\end{abstract}

\setcounter{footnote}{0}
\section{Introduction}
\paragraph{Background.}
Property testing~\cite{Goldreich17,GR00} is the task of (approximately) distinguishing objects having some specific property from those which are ``far" from having it, without necessarily looking at the objects in their entirety. An interesting subfield is discrete distribution testing~\cite{Canonne15}, where the objects are probability distributions.

One of the main tasks in (discrete) distribution testing, namely closeness testing, is about deciding 
whether two distributions $p$ and $q$ over $[n]$ are equal or $\epsilon$-far from each other in the $l_1$-norm, 
given access only to a limited number of samples of each distribution. Early testers~\cite{BFRSW00,BFRSW13} used a method based on collisions. Using improved estimators, testers with optimal sample complexity have then been constructed~\cite{CDVV14,DK16}. 

Very recently, Andoni, Malkin and Nosatzki \cite{AMN19} have, for the first time, considered distribution testing in the two-party setting. Here two players, Alice and Bob, each own as input $t$ samples of the distributions $p$ and $q$: Alice has $t$ samples from $p$ and Bob has $t$ samples from $q$. The goal is for Alice and Bob to decide if the two distributions are equal or $\epsilon$-far from each other in the $l_1$-norm, using as little communication as possible. By adapting the techniques from prior works \cite{CDVV14,DK16}, Andoni, Malkin and Nosatzki have shown that this problem (named \CT{} in \cite{AMN19}) can be solved with high probability using $O\big(\frac{n^2}{t^2\epsilon^4}+1\big)$ bits of communication whenever $t$ is above the information-theoretic lower bound (given in Equation (\ref{ic}) below) which is the minimum number of samples needed so that meaningful information about $p$ and $q$ can be extracted from them.
They also showed a matching lower bound $\Omega\big(\frac{n^2}{t^2}\big)$ on the two-party communication complexity of \CTc{}.


\paragraph{Our results.}
In this paper we investigate the quantum communication complexity of this problem. Our main result shows that a significant advantage can be obtained in the quantum setting when at least one of the two distributions has low $l_2$-norm. Concretely, for any $\gamma>0$, we consider the version of \CT{} in which the inputs $p,q$ satisfy the condition $\min(||p||_2, ||q||_2)\le \gamma t\epsilon^2/ n$. We denote this problem \CTs{}. The lower bound from \cite{AMN19} shows that for $\gamma=\Omega(\frac{1}{\sqrt{\log n}})$,
this version is as hard as the original version of the problem.\footnote{Indeed, the lower bound $\Omega\big(\frac{n^2}{t^2}\big)$ from \cite{AMN19} is  shown for input distributions such that  $||p||_2=||q||_2=O (\frac{t\epsilon^2}{n\sqrt{\log n}})$.}
As in all previous works \cite{AMN19,CDVV14,DK16}, we will assume throughout the paper that $t$ is above the information-theoretic threshold:
\begin{equation}\label{ic}
 t\geq C\max(n^{2/3}\cdot\epsilon^{-4/3},\sqrt{n}\cdot\epsilon^{-2}),  
\end{equation}
where $C$ is a universal constant (see \cite{AMN19,CDVV14} for details).

We first show the following theorem.
\begin{theorem}\label{th:main}
There exists an absolute constant $\gamma_0$ such that the following holds: for all $\gamma\le \gamma_0$, 
the problem \CTs{} can be solved with high probability by a quantum protocol that uses $\tilde O\big(\frac{n}{t\epsilon^2}+1\big)$ qubits of communication.
\end{theorem}
Theorem \ref{th:main} shows that a significant advantage (a quadratic improvement in the communication complexity) can be obtained in the quantum setting when at least one of the two distributions has low $l_2$-norm. 

We also obtain the following lower bound, which shows that the upper bound of Theorem \ref{th:main} is optimal, even for $\epsilon=1/2$.

\begin{theorem}\label{th:lb}
There exists an absolute constant $c>1$ such that the following statement holds for any value $t\le n/\log^c n$ and any $\gamma=\Omega(1/\sqrt{\log n})$: any quantum protocol that solves \CTsc{} with high probability requires $\tilde \Omega(n/t)$ qubits of communication. 
\end{theorem}
Precisely, the lower bound holds for any $\gamma$ greater than some $\gamma_{LW}=O(\frac{1}{\sqrt{\log n}})$. 
Therefore there is a regime where $\gamma_{LW}\leq \gamma_0$, for which our upper bound is then tight.

\paragraph{Overview of our main techniques.}
Our upper bound (Theorem \ref{th:main}) is obtained by following the framework used in \cite{AMN19}, which relies on the estimator from~\cite{CDVV14} for the $l_2$-distance. Indeed, as suggested by~\cite{BFRSW00,CDVV14} and then extensively studied in~\cite{DK16}, there is a reduction from closeness testing in the $l_1$-distance to closeness testing in the $l_2$-distance. The efficiency of the reduction, however, depends on the $l_2$-norm of the distribution.
The protocol in~\cite{AMN19} thus proceeds in two steps. In the first step, Bob shares some information with Alice about the observed shape of his distribution, so that they can recast their distributions into two distributions $p'$ and $q'$ that have smaller $l_2$-norm while preserving the $l_1$-distance (i.e., $||p'-q'||_1=||p-q||_1$). In the second step, Alice and Bob use the reduction to closeness testing in the $l_2$-distance mentioned and implement the estimator of \cite{CDVV14} in the two-party setting. This estimator requires estimating with good precision the $l_2$-distance between two vectors. This is done by using a two-party implementation of the sketching method by Alon, Matias and Szegedy \cite{AMS99}. 

For the case of low-norm distribution (more precisely, when considering the problem \CTs{} with $\gamma$ constant), the first step is unnecessary: the two distributions already have  a low enough $l_2$-norm, so that the reduction to closeness testing in the $l_2$-distance can be used without any preprocessing. We thus only need to show how to implement the second step from \cite{AMN19} more efficiently using quantum communication. The key idea is to use the quantum algorithm by Montanaro~\cite{Montanaro16} which gives a quadratic speedup over the classical sketching method from \cite{AMS99} in the query complexity model. We show how to adapt this quantum algorithm to the two-party setting in Section \ref{sub:qAMS}. 

Our lower bound (Theorem \ref{th:lb}) first applies the same reduction as in \cite{AMN19}, which reduces some specific version of the Gap-Hamming distance to \CTsc{}. In the classical case, \cite{AMN19} showed that the communication complexity of that version of the Gap-Hamming distance is $\tilde \Omega((n/t)^2)$ bits. Our main technical contribution proves that the quantum communication complexity of this problem is $\tilde \Omega(n/t)$. We use the pattern matrix method \cite{SherstovSICOMP11}. To obtain our lower bound, we show that the pattern matrix method, which is generally formulated only for total functions, can be generalized to partial functions. 

\paragraph{Relation with known quantum advantages in testing and learning.}
Several quantum algorithms have been designed for property testing and learning theory (we refer to \cite{Arunachalam+17} and \cite{Montanaro+16} for excellent surveys of these fields).  In most settings considered so far in quantum learning theory, however, the quantum algorithms crucially exploit the fact that the data can be accessed in a quantum way (e.g., we can query a quantum superposition of samples), which makes it difficult to directly compare the performance of quantum algorithms with the performance of classical methods (which can only access the data in a classical way). The results of the present paper show a quantum advantage, in terms of communication cost, for the setting where both classical and quantum protocols can access the data in the same way -- the input is given as a set of classical samples.

\paragraph{Open problem.} An intriguing question is whether a similar quantum advantage is achievable when both input distributions have higher $l_2$-norm, i.e., whether the upper bound we obtain in Theorem \ref{th:main} holds not only for constant~$\gamma$ but also for larger values of $\gamma$. Currently, we do not know how to improve the complexity of the first part of the classical protocol from \cite{AMN19}, which (as mentioned above) converts the input distributions into distributions of sufficiently small $\ell_2$-norm, using quantum communication. We left this question as an open problem.  
\section{Preliminaries}

\subsection{Definitions and Notations}

A typical communication task for Alice and Bob is to compute (sometimes only with some probability of success) a function $f$ on some inputs $x,y$ where $x$ is given to Alice and $y$ to Bob.
A communication protocol is an algorithmic description of message sending between Alice and Bob that solves the task for any possible pair of inputs.
The communication complexity~\cite{KN97} of such function is the minimum required numbers of bits the most efficient protocol solving the task must exchange in the worst case (regarding inputs).

The quantum communication complexity~\cite{Wolf02,Brassard04} of a function is the equivalent using qubits instead of bits.
Qubits correspond to elements of some Hilbert space of dimension $2$. We will use the bra-ket notation $\ket{\phi}_{R}$ to denote a qubit (and by extension an $n$-qubit string) $\phi$ of a register $R$.

As described in the introduction, Alice's input consists of $t$ samples from a discrete distribution $p\colon [n]\to (0,1)$. Bob's input consists of $t$ samples from a discrete distribution $q\colon [n]\to (0,1)$.
We call $X_i$ (resp.~$Y_i$) the number of samples of Alice (resp.~Bob) corresponding to element $i$. We call $X,Y\in [t]^n$ the occurrence vectors of Alice and Bob, i.e., the vectors with $i$-th coordinate $X_i$ and $Y_i$, respectively.

For a vector $x\in\mathbb{R}^n$, we will denote by $||\cdot||_1$ the $l_1$-norm, which is defined as $||x||_1=\sum_i^n |x_i|$, and denote by $||\cdot||_2$ the $l_2$-norm, which is defined as $||x||_2=\sqrt{\sum_i^n x_i^2}$.
We use $\tilde{O},\tilde{\Omega}$ instead of $O,\Omega$ when we neglect factors of logarithmic order in the parameters of the problem ($n,t,\epsilon$).
We denote by $Poi(\lambda)$ the Poisson distribution with parameter $\lambda\in\mathbb{N}$.
%
\subsection{CDVV~\cite{CDVV14} Estimator}
The idea behind the estimator from~\cite{CDVV14} is similar to estimation using collisions~\cite{BFRSW13}, except it assumes Poisson sampling for getting some independence that simplifies the analysis, and therefore needs some corrective terms to shift the mean so that the estimation is unbiased. It is defined by first drawing a number $M\sim{Poi(\lambda)}$, 
where $\lambda$ is some parameter high enough, and then taking $M$ arbitrary samples on each side and computing using the occurrence vectors $X,Y$ of those samples:

\begin{align*}
    Z&=\frac{\sqrt{\sum_i ((X_i-Y_i)^2-X_i-Y_i)}}{M}\\
    &= \frac{\sqrt{\sum_i (X_i-Y_i)^2  -2M}}{M}.
\end{align*}

\subsection{Classical Protocol}
We now describe the protocol from \cite{AMN19} for the case of small $l_2$-norm. 

The main idea behind the protocol is based on the CDVV estimator described in the previous subsection, combined with a reduction from $l_1$-distance estimation to $l_2$-distance estimation.
Considering as well the errors of the estimator, after rescaling and shifting from $Z$, Andoni, Malkin and Nosatzki found that it is enough to compare some approximation of the term $\Delta=||X-Y||_2^2$ to some threshold $\tau$ to distinguish the two cases. As discussed in the introduction, the original protocol from \cite{AMN19} has a communication step to recast the probability distributions into ones with smaller $l_2$-norms. Since this step is not necessary for the case of distributions with small $l_2$-norms, we omit it in the following description.

\begin{algorithm}[H]
   \caption{\bf Classical Protocol for Distribution Closeness Testing from \cite{AMN19}}\label{protocol}
\begin{algorithmic}[1]
\State Fix $\alpha=\Theta(\frac{t\epsilon^2}{n}+1)$;
\State Alice and Bob each estimate $||p||_2$ and $||q||_2$ up to a factor 2; if the two estimates are not within a factor 4, output ``$\epsilon$-FAR'';\label{alg:line:check}
\State Alice and Bob approximate $\Delta=||X-Y||_2^2$ up to a $(1+\alpha)$ factor using standard techniques;
\State If $\Delta$ is less than $\tau=\frac{\epsilon^2t^2}{2n}+2t$ output ``SAME", and otherwise output ``$\epsilon$-FAR'';
\end{algorithmic}
\end{algorithm}



The analysis from {\cite{AMN19}, in particular Lemma 6, shows the correctness of the protocol. More precisely, the following statement can be obtained for the case of input distributions with low $l_2$-norms.
\begin{theorem}{\cite{AMN19}}\label{th:amn19UB}
There exists an absolute constant $\gamma_0$ such that the following holds: for any input distributions $p$ and $q$ such that $\min(||p||_2, ||q||_2)\le \gamma_0 t\epsilon^2/ n$,
the above protocol correctly distinguish between the case $p=q$ and the case $||p-q||_1\ge \epsilon$ with probability at least 2/3.
\end{theorem}
The communication complexity is dominated by the third step, which requires $\tilde{O}\left(1/\alpha^2\right)=\tilde{O}\left(n^2/(t^2\epsilon^4)+1\right)$ bits.

\section{Quantum Protocol}
In this section we describe our quantum protocol for the problem \CTs, and prove Theorem \ref{th:main}.
\subsection{Description of the Whole Protocol}
\vspace{-4mm}
\begin{algorithm}[H]
   \caption{\bf Quantum Protocol for Distribution Closeness Testing}\label{qprotocol}
\begin{algorithmic}[1]
\State Fix $\alpha=\Theta(\frac{t\epsilon^2}{n}+1)$;
\State Alice and Bob each estimate $||p||_2$ and $||q||_2$ up to a factor 2; if the two estimates are not within a factor 4, output ``$\epsilon$-FAR'';
        \State Alice and Bob approximate $\Delta=||X-Y||_2^2$ up to a $(1+\alpha)$ factor using the procedure of Section $\ref{sub:qAMS}$;
        \State If $\Delta$ is less than $\tau=\frac{\epsilon^2t^2}{2n}+2t$ output ``SAME", and otherwise output ``$\epsilon$-FAR'';
\end{algorithmic}
\end{algorithm}
The communication complexity is again dominated by the third step, which requires only $\tilde{O}\left(1/\alpha\right)=\tilde{O}\left(n/(t\epsilon^2)+1\right)$ qubits, as described in the next subsection. The correctness is guaranteed by the analysis of Theorem \ref{th:amn19UB}, since the quantum protocol performs identical calculations as the classical protocol. 
This proves Theorem \ref{th:main}.
 
\subsection{Montanaro Approximation}\label{sub:qAMS}
The classical protocol \cite{AMN19} uses standard techniques, such as the AMS algorithm~\cite{AMS99}, in order to approximate $\Delta=||X-Y||_2^2$. 
The AMS algorithm uses a family $\mathcal{H}$ of $O(n^2)$ hash functions $h_i : [n]\rightarrow \{-1,+1\}$ that are 4-wise independent.
Given a list of numbers $l=(l_1.\dots,l_n)$, the AMS algorithm gives an estimate of 
$||l||_2^2$ by computing random estimates $f(i)$ with the following subroutine many times and taking the median of the results.\footnote{The idea behind the AMS algorithm is that by developing the square, the ``crossed" product terms' influence should vanish, i.e.,
$\mathbb{E}[h_i(j)l_j\cdot h_i(j')l_{j'}]=0$ for $j\neq j'$ because then $\mathbb{E}[h_i(j)\cdot h_i(j')]=0$, while the terms 
$h_i(j)l_j\cdot h_i(j)l_{j}=l_j^2$ will always stay.}

\begin{algorithm}[H]
   \caption{\bf Alon-Matias-Szegedy algorithmic subroutine}\label{AMS}
\begin{algorithmic}[1]
\State Draw a random index $i$ to choose a hash function $h_i\in\mathcal{H}$;
\State Return $f(i)=(\sum_{j=1}^n h_i(j)\cdot l_j)^2$;
\end{algorithmic}
\end{algorithm}\vspace{-3mm}

\noindent In the classical setting, this subroutine has to be repeated $\tilde{O}(1/\alpha^2)$ times to get a $(1+\alpha)$-approximation.

Montanaro showed how to achieve the same approximation quantumly using this subroutine only $\tilde{O}(1/\alpha)$ times (see Theorems 12 and 14 in \cite{Montanaro16}). The subroutine, however, needs to be called in superposition, i.e., Montanaro's approach requires a quantum oracle $O_f$ that performs the following map:
$$O_f: \ket{i}\ket{y}\rightarrow \ket{i}\ket{y+f(i)}.$$
It also requires access to its inverse $O_f^{-1}$.

In our communication setting, we want to use this approach with $l:=X-Y$.  A difficulty is that the data is split between the two parties: only Alice knows~$X$ and only Bob knows $Y$.  We now explain how to overcome this difficulty.
For a particular index $i$, Alice can compute $\sigma_a(i) = \sum_{j=1}^n h_i(j) X_j$, 
and then transmit it to Bob.
Bob can similarly do his own computation $\sigma_b(i) = \sum_{j=1}^n h_i(j) Y_j$, then subtract $\sigma_a(i)$ and square in order to get 
\begin{equation}\label{eq:diff}
(\sigma_a(i)-\sigma_b(i))^2 = 
\left(\sum_{j=1}^n h_i(j)\cdot (X_j-Y_j)\right)^2=f(i).
\end{equation}
We describe below our implementation of the oracle $O_f$, based on these ideas. In the following description, Bob's input is the quantum state $\ket{i}_{I}\ket{y}_{Y}$ for some bit-strings $i$ and $y$, where $I$ and $Y$ are two quantum registers. 

\begin{algorithm}[H]
   \caption{\bf Protocol for building the oracle 
   $O_f$}\label{OracleProtocol}
\begin{algorithmic}[1]
\State Bob maps $\ket{i}_{I}\ket{y}_{Y}$ to $\ket{i}_{I}\ket{y}_{Y}\ket{\sigma_b(i)}_{B}$ and sends register $I$ to Alice.
\State Alice creates another register $A$, computes the value $\sigma_a(i)$ in it, and sends the two registers $A$ and $I$ to Bob. At the end of this step, Bob thus owns 
\[
\ket{i}_{I}\ket{y}_{Y}\ket{\sigma_a(i)}_{A}\ket{\sigma_b(i)}_{B}.
\]
\State Bob performs the computation of Equation (\ref{eq:diff}) using the arguments from $A$ and $B$, to get the state
\[
\ket{i}_{I}\ket{y+f(i)}_{Y}\ket{ \sigma_a(i)}_{A}\ket{ \sigma_b(i)}_{B}.
\]
\State Bob erases the contents of register $B$, and sends registers $A$ and $I$ to Alice.
\State Alice erases the contents of register $A$ and sends back register $I$ to Bob.
\end{algorithmic}
\end{algorithm}

Note that it is the linear property of the AMS computation (more precisely, Equation (\ref{eq:diff})) that allows the approximation to be computed even when the data are shared by several parties. Transmitting register $I$ requires $O(\log n)$ qubits. Transmitting register $A$ requires $O(\log t)$ qubits. The overall communication complexity of the oracle protocol is thus $O(\log n + \log t)$. 

An implementation of the inverse $O_f^{-1}$ can be obtained similarly, by subtracting instead of adding the value of $f(i)$ at Step 3. 
We can thus apply Montanaro's algorithm \cite{Montanaro16}, simply by replacing each oracle query in Montanaro's algorithm by our distributed implementation of $O_f$ (or $O_f^{-1}$). This enables us to obtain, with high probability, an $(1+\alpha)$-approximation of $||l||_2^2=\Delta$ using 
\[
\tilde O((1/\alpha)(\log n + \log t))=\tilde O(1/\alpha)
\]
qubits of communication, as claimed.
\section{Quantum Lower Bound}\label{sec:LW}
In this section we prove Theorem \ref{th:lb}.

\subsection{Hamming Reduction}
For bit-strings $x,y\in \{0,1\}^n$, we denote by $x\cap y$ the set of indices where $x$ and~$y$ both have a one, and  write $|x\cap y|$ for its size.
Note that for bit-strings the $l_1$-norm corresponds to Hamming weight, i.e., the number of ones in the string, and the $l_1$-distance corresponds to the Hamming distance.

In the classical communication setting, Andoni, Malkin and Nosatzki~\cite{AMN19} proved a lower bound for some closeness testing problem by considering the following problem involving the Hamming distance of two binary strings.

\begin{center}
\fbox{
\begin{minipage}{11 cm} 
Let $n\geq 1$ be a multiple of $4$. Let $\beta=\beta(n)=\sqrt{n/32}$,  and $\kappa>1$.   
With probability at least $0.9$, for $x,y\in\{0,1\}^n$ with $||x||_1=||y||_1=n/2$, 
distinguish between the case where $||x-y||_1=n/2$ versus $||x-y||_1-n/2\in[\beta,\kappa\beta]$. 
\end{minipage}
}
\end{center}

Notice that if $||x||_1=||y||_1=n/2$ then the equality 
\[
||x-y||_1=2(n/2-|x\cap y|)=n-2|x\cap y|
\]
holds. 
The above problem then can be reformulated as the following communication problem that we call PromisedGHD($n,\kappa$): 


\begin{center}
\fbox{
\begin{minipage}{11 cm} 
PromisedGHD($n,\kappa$)\\
For $x,y\in\{0,1\}^n$, where $n$ is a multiple of $4$, with $||x||_1=||y||_1=n/2$, $\beta=\beta(n)=\sqrt{n/32}$,  and $\kappa>1$, with probability at least $0.9$
distinguish between $|x\cap y|=\frac{n}{4}$ 
and $|x\cap y|\in[\frac{n}{4}-\frac{\kappa\beta}{2},\frac{n}{4}-\frac{\beta}{2}]$.
\end{minipage}
}
\end{center}
 
 Particularly, the reduction of~\cite{AMN19} only concerns the regime where  $\kappa=O(\sqrt{\log n})$.
 
The main result of this subsection is:
\begin{theorem}\label{th:PGHDlw}
The quantum communication complexity of
PromisedGHD($n,\kappa$) 
is $\Omega(\sqrt{n})$.
\end{theorem}

In order to prove Theorem \ref{th:PGHDlw}, we will consider a similar problem defined on smaller inputs. More precisely, define the following problem SmallPGHD($n',g$):

\begin{center}
\fbox{
\begin{minipage}{11 cm} 
SmallPGHD($n',g$)\\
For $x',y'\in\{0,1\}^{n'}$ where $n'$ is a multiple of $4$, with $||x' ||_1=||y' ||_1=n'/2$, with probability at least $2/3$
distinguish between $|x'\cap y'|=n'/4$ 
and $|x'\cap y'|\in[n'/4-g,n'/4-1]$
\end{minipage}
}
\end{center}

We first show a reduction from SmallPGHD($n',g$) to PromisedGHD($n,\kappa$) for appropriate parameters. 

%

\begin{lemma}
For any 
$g\leq\frac{\kappa}{2\sqrt{32}}$, 
SmallPGHD($n',g$) reduces to PromisedGHD($n'^2,\kappa$).
\end{lemma}

\begin{proof}
Assume $x',y'$ are inputs of size $n'$ to the SmallPGHD($n',g$) problem. By repeating them $n'$ times in a padding fashion, we build inputs $x,y$ for the PromisedGHD($n,\kappa$) problem where $n=n'^2$ and therefore $\beta=\frac{n'}{\sqrt{32}}$. 
If $|x'\cap y'|=n'/4$, we have $|x\cap y|=n'^2/4=n/4$.
Otherwise, if $|x'\cap y'|\in[\frac{n'}{4}-g,\frac{n'}{4}-1]$,  we have $|x\cap y|\in[\frac{n}{4}-\frac{\kappa\beta}{2},\frac{n}{4}-\frac{\beta}{2}]$ as long as $g\leq\frac{\kappa}{2\sqrt{32}}$.
Indeed, for the upper part of the interval:
\[
 \frac{n}{4}-\frac{\beta}{2} =\frac{n'^2}{4}-\frac{n'}{2\sqrt{32}}
  \geq n'(\frac{n'}{4}-1)
\]
while for the lower part, we need:
\[
\frac{n}{4}-\frac{\kappa\beta}{2} =\frac{n'^2}{4}-\frac{\kappa n'}{2\sqrt{32}}
\leq n'(\frac{n'}{4}-g),
\]
which is equivalent to the bound on $g$ mentioned just above.
\end{proof}

The hardness of SmallPGHD($n',g$) will follow from the hardness of a more restricted problem stated in Theorem~\ref{thm:sherstov}.
The latter theorem is the main technical contribution of this section, and we devote Section~\ref{sec:mainTechnical} to its proof.

%
%

\subsection{Main Technical Result}
\label{sec:mainTechnical}

We will model 
partial functions as mappings $f\colon X\to
\Re\cup\{*\}$, where $X$ is a finite set and 
the range element $*$ represents undefined output.
We let $\dom f= \{x: f(x)\ne *\}.$
It will be convenient here to represent the Boolean
values ``true'' and
``false'' by $-1$ and $+1,$ respectively. This
departure from the classical representation (of using $0$
and $1$) has no effect on quantum
communication complexity. For a communication problem
$F\colon X\times Y\to\{-1,+1,*\},$ we let $Q^*_\epsilon(F)$ denote
the $\epsilon$-error quantum communication complexity
of $F$ with arbitrary prior entanglement. Note that an
$\epsilon$-error protocol for $F$ is allowed to behave
arbitrarily on inputs outside $\dom F.$

We are now ready to prove the main technical result used in the proof of Theorem~\ref{th:lb}.

\begin{theorem}
\label{thm:sherstov}
Let $n$ be an integer divisible by $4.$ Consider the
partial communication problem
$F_n\colon\zoon\times\zoon\to\{-1,+1,*\}$ given by 
\begin{align*}
F_n(x,y)=\begin{cases}
-1 &\text{if $|x|=|y|=n/2$ and $|x\cap y|=n/4,$}\\
+1 &\text{if $|x|=|y|=n/2$ and $|x\cap y|=n/4-1,$}\\
* &\text{otherwise.}
\end{cases}
\end{align*}
Then
$Q_{1/3}^*(F_n)=\Omega(n).$
\end{theorem}

The remaining part of this section is devoted to the proof of this theorem.
We start by reviewing relevant background on the
pattern matrix method~\cite{SherstovSICOMP11} for
quantum communication lower bounds.
Let $k$ and $n$ be positive integers, where $k<n$ and $k\mid n.$ Partition
$[n]$ into $k$ contiguous blocks, each with $n/k$ elements:
\[
[n] = \left\{1,  2,  \dots,   \frac{n}{k}\right\} 
       \cup 
	   \left\{\frac{n}{k}+1,   \dots,  \frac{2n}{k}\right\}
         \cup \cdots \cup
         \left\{\frac{(k-1)n}{k}+1,  \dots,  n\right\}.
\]
Let $\VV(n,k)$ denote the family of subsets $V\subseteq[n]$
that have exactly one element in each of these blocks (in particular,
$|V|=k$).  Clearly, $|\VV(n,k)|=(n/k)^k.$ For a bit string $x\in\zoon$
and a set $V\in\VV(n,k),$ define the \emph{projection of $x$ onto
$V$} by
$x|_V  =      (x_{i_1},x_{i_2},\dots,x_{i_k}) \in
\zoo^k,$
where $i_1<i_2<\cdots< i_k$ are the elements of $V.$ 
For $\phi\colon \zoo^k\to\Re\cup\{*\},$ the \emph{$(n,k,\phi)$-pattern matrix} is
the matrix $A$ given by
\[ A = \Big[\phi(x|_V\oplus
w)\Big]_{x\in\zoon,\,(V,w)\in\VV(n,k)\times\zoo^k}  \;. \]
In words, $A$ is the matrix of size $2^n$~by~$(n/k)^k2^k$ whose rows are
indexed by strings $x\in\zoon,$ whose columns are indexed by pairs
$(V,w)\in\VV(n,k)\times\zoo^k,$ and whose entries are given by
$A_{x,(V,w)}= \phi(x|_V\oplus w).$

The pattern matrix method gives a lower bound on the
quantum communication complexity of a pattern matrix in
terms of the \emph{approximate degree} of its generating
function. We now define this notion formally. 
Let $f\colon X\to\Re$ be given, for a finite subset $X\subset\Re^{n}.$
The \emph{$\epsilon$-approximate degree} of $f,$ denoted $\degeps(f),$
is the least degree of a real polynomial $\pi$ such that
$|f(x)-\pi(x)|\leq \epsilon$ for all $x\in X.$
One generalizes this definition to partial functions
$f\colon X\to\Re\cup \{*\}$ 
by letting $\degeps(f)$ be the least degree of a real polynomial
$\pi$ with 
\begin{align*}
 & |f(x)-\pi(x)|\leq\epsilon, &  & x\in\dom f,\\
 & |\pi(x)|\leq1+\epsilon, &  & x\in X\setminus\dom f.
\end{align*}
We will need the following version of the pattern
matrix method for quantum lower bounds.

\begin{theorem}
\label{thm:pmm}
Let $F$ be the $(n,k,f)$-pattern matrix, 
where $f\colon \zoo^k\to\{-1,+1,*\}$ is given.
Then for every $\epsilon\in[0,1)$ and every $\delta<\epsilon/2,$
\begin{align*}
Q^*_{\delta}(F) &\geq
\frac{1}{4}  \degeps(f)\log \left(\frac{n}{k}\right) - 
\frac12 \log\left(\frac{3}{\epsilon-2\delta}\right).
\end{align*}
\end{theorem}
\noindent
Theorem~\ref{thm:pmm} is a generalization of the
original pattern matrix method
of~\cite{SherstovSICOMP11} to partial functions. For
the reader's convenience, we give a detailed proof of
Theorem~\ref{thm:pmm} in Appendix~\ref{sec:pmm-proof}.\vspace{2mm}

\noindent{\emph{Proof of Theorem \ref{thm:sherstov}}.
The communication complexity of $F_n$ is 
monotone in $n,$ due to 
$F_{n}(x,y)=F_{n+4}(x0011,y0101).$ As a result, it
suffices to prove the theorem for $n$ divisible by $3.$ 
Under this divisibility assumption, define $k=n/6$ and 
consider the function
$\PMAJ_k\colon\zoo^k\to\{-1,+1,*\}$ given by
\begin{align*}
\PMAJ_k(x)=
\begin{cases}
-1 &\text{if $|x|=k/2,$}\\
+1 &\text{if $|x|=k/2-1,$}\\
* &\text{otherwise.}
\end{cases}
\end{align*}
Let $P$ be the $(2k,k,\PMAJ_k)$-pattern matrix.
It is a well-known 
fact~\cite{paturi92approx,bun-thaler13and-or-tree}
that $\deg_{1/3}(\PMAJ_k)=\Omega(k).$
As a result, Theorem~\ref{thm:pmm}
implies that $Q^*_{1/7}(P)=\Omega(k)$ and hence also
$Q^*_{1/3}(P)=\Omega(k).$

Writing
$P=[\PMAJ_k((x_1\overline x_1 x_2\overline x_2\ldots
x_{2k}\overline
x_{2k})|_V)]_{x\in\zoo^{2k},V\in\VV(4k,k)}$
makes it clear that $P$ is a restriction of the 
more general communication problem 
$G\colon\zoo^{4k}\times\zoo^{4k}\to\{-1,+1,*\}$ defined
by
\begin{align*}
G(x,y)=\begin{cases}
-1&\text{if $|x|=2k,\;|y|=k,$ and $|x\cap y|=k/2,$}\\
+1&\text{if $|x|=2k,\;|y|=k,$ and $|x\cap y|=k/2-1,$}\\
*&\text{otherwise.}
\end{cases}
\end{align*}
As a result, $Q^*_{1/3}(G)\geq Q^*_{1/3}(P)=\Omega(k).$
This in turn implies that
$Q^*_{1/3}(F_n)=\Omega(k)$ because 
$G(x,y)=F_n(x1^k0^k,y1^{2k})$.

\subsection{Closeness Testing Reduction}
In this subsection we explain how the lower bound on the quantum communication complexity of PromisedGHD($n,\kappa$) (Theorem \ref{th:PGHDlw}) implies Theorem~\ref{th:lb}.\vspace{2mm}



\noindent{\emph{Proof of Theorem \ref{th:lb}}.
In~\cite{AMN19}, Andoni, Malkin and Nosatzki show a reduction from the problem PromisedGHD($m,\kappa$) to the problem \CTc{} with parameters $m=\frac{n^2}{t^2\log^3 n}$ and $\kappa=O(\sqrt{\log n})$.\footnote{The reduction is actually stated, in Theorem 9 in \cite{AMN19}, as a reduction from PromisedGHD($m,\kappa$) to the variant of \CTc{} where the number of samples is Poi($t$) instead of exactly $t$. Nevertheless the Poisson version easily reduces to the original problem with $10t$ number of samples if we allow some extra error, since by Chebyshev's inequality the probability that the Poisson version gives more than $10t$ samples is less than $1/81$.} Theorem \ref{th:PGHDlw} thus gives us the claimed lower bound. In the remaining of the proof, we show that the distributions used to prove the lower bound have low $l_2$-norm.


The input distributions of \CTc{} used in the reduction shown in \cite{AMN19}, which we will denote $a$ and $b$, are of the following form:
half of the mass is uniformly distributed on $d=n/10$ elements, and the other half of the mass on $l=C_0 \cdot t \cdot  \log n$ other elements, where $C_0$ is some constant.
Therefore:
\begin{align*}
||a||_2=||b||_2&=\sqrt{d(\frac{1}{2d})^2+l(\frac{1}{2l})^2} \\
&=\frac{1}{2}\sqrt{\frac{1}{d}+\frac{1}{l}} \\
&=\frac{1}{2}\sqrt{\frac{10}{n}+\frac{1}{C_0 t \log n}}\\
&\leq \frac{1}{2}\sqrt{(10+\frac{1}{C_0}) \frac{1}{t \log n}},
\end{align*}
since  $\frac{1}{n}\leq \frac{1}{t \log n}$ because $t\leq \frac{n}{\log^c n}$.

Define $\gamma_{LW}$ as the smallest $\gamma$ such that $\min(||a||_2, ||b||_2)\le \gamma t\epsilon^2/ n$ holds. The above calculations  show that
\[
\gamma_{LW}=\frac{||a||_2 n}{t\epsilon^2}
\leq  \frac{1}{2}\sqrt{10+\frac{1}{C_0}} \frac{n}{t\epsilon^2\sqrt{t\log n}}
\leq  \frac{1}{2C^{3/2}}\sqrt{10+\frac{1}{C_0}} \frac{1}{\sqrt{\log n}}
\]
because $t\geq C n^{2/3}\cdot\epsilon^{-4/3}$ by~(\ref{ic}). Thus $\gamma_{LW}=O(1/\sqrt{\log n})$.
This concludes the proof. 

\bigskip \noindent\textbf{Acknowledgements.}
Guillaume Malod was partially supported by a JSPS Invitational Fellowships for Research in Japan.
Aleksandrs Belovs is supported by the ERDF project number 1.1.1.2/I/16/113.
Arturo Castellanos is grateful to Shin-ichi Minato for his support, and also to MEXT.
Alexander A. Sherstov\ was supported by NSF grant CCF-1814947.
Fran{\c c}ois Le Gall was supported by JSPS KAKENHI grants Nos.~JP16H01705, JP19H04066, JP20H00579, JP20H04139 and by the MEXT Quantum Leap Flagship Program (MEXT Q-LEAP) grant No.~JPMXS0118067394.

\bibliographystyle{alpha}
\bibliography{refs,sherstov}

\appendix

\section{The Pattern Matrix Method for Partial Functions}
\label{sec:pmm-proof}

The purpose of this appendix is to provide a detailed
proof of Theorem~\ref{thm:pmm} for partial functions.
Our proof closely follows the original proof of the
pattern matrix method in~\cite{SherstovSICOMP11},
developed there for total functions.

We start by recalling the Fourier transform 
for functions $f\colon\zoon\to\Re.$ 
For $S\subseteq\oneton,$
define $\chi_{S}\colon\zoon\to\moo$ by $\chi_{S}(x)=(-1)^{\sum_{i\in S}x_{i}}.$
Then every function $f\colon\zoon\to\Re$ has a unique representation
of the form 
\[
f=\sum_{S\subseteq\oneton}\hat{f}(S)\,\chi_{S},
\]
where $\hat{f}(S)=2^{-n}\sum_{x\in\zoon}f(x)\chi_{S}(x).$ The reals
$\hat{f}(S)$ are called the \emph{Fourier coefficients of $f.$}

For a real matrix $A,$ we let $\|A\|_1$ denote the sum
of the absolute values of the entries of $A.$ We 
let $\|A\|$ denote the spectral norm of $A.$ Recall that 
$\|A\|= \max_{x:\|x\|_2=1} \|Ax\|_2.$ 
The following theorem~\cite[Theorem~4.3]{SherstovSICOMP11} 
determines the spectral norm of a pattern matrix in
terms of the Fourier spectrum of its generating
function.

\begin{theorem}
\label{thm:pattern-spectrum}
Let $\phi\colon \zoo^k\to\Re$ be given. Let $A$ be the $(n,k,\phi)$-pattern
matrix. Then 
\[ 
\|A\| \;=\; \sqrt{ 2^{n+k} \left( \frac{n}{k}\right)^k} 
\;\max_{S\subseteq[k]} 
    \left\{ |\hat\phi(S)| \left( \frac{k}{n}\right)^{|S|/2} \right\}.
\]
\end{theorem}

We will also need the following dual characterization of approximate
degree of partial functions, analogous to the dual
characterization for total functions used in~\cite{SherstovSICOMP11}.
\begin{theorem}
\label{thm:dual-approx} Let
$f\colon\zoon\to\Re\cup\{*\}$ be a given function, $d\geq0$ an
integer.
Then $\degeps(f)>d$ if and only if there exists
$\psi\colon\zoon\to\Re$ such that 
\[
\sum_{x\in\dom f}f(x)\psi(x)-\sum_{x\notin\dom f}|\psi(x)|-\epsilon\|\psi\|_{1}>0,
\]
and $\hat{\psi}(S)=0$ for $|S|\leq d.$ 
\end{theorem}

\noindent
Theorem~\ref{thm:dual-approx} follows from linear programming
duality; see~\cite{SherstovSICOMP11,sherstov11quantum-sdpt} for
details.

Next, we derive a version of the \emph{generalized
discrepancy method} for partial functions, by adapting
the analogous proof in~\cite[Theorem~2.8]{SherstovSICOMP11} 
for total functions.

\begin{theorem}
Let $X,Y$ be finite sets and $F\colon X\times
Y\to\{-1,+1,*\}$ a given function.
Let $\Psi=[\Psi_{xy}]_{x\in X,\,y\in Y}$ be any real matrix with $\|\Psi\|_1=1.$
Then for each $\epsilon>0,$
\[ 4^{Q^*_\epsilon(F)} \geq
\frac{1}{3\,\|\Psi\|\sqrt{|X|\,|Y|}}
\left(\sum_{(x,y)\in\dom F}\Psi_{x,y}F(x,y)-
\sum_{(x,y)\notin\dom F}|\Psi_{x,y}| - 
2\epsilon\right).
\]
\label{thm:discrepancy-method}
\end{theorem}

\begin{proof}
Let $P$ be a quantum protocol with prior entanglement that computes
$F$ with error $\epsilon$ and cost $C.$ 
Let $\Pi$ be the matrix of acceptance probabilities of
$P,$ so that $\Pi_{x,y}$ is the probability that $P$
accepts the input $(x,y).$ 
It is shown in the 
proof of~\cite[Theorem~2.8]{SherstovSICOMP11} that
\begin{align}
\sum_{x\in X}\sum_{y\in Y}  \Psi_{x,y}(1-2\Pi_{x,y})
  &\leq \|\Psi\|\; \left(2\cdot 4^{C}+1\right)\sqrt{|X|\;|Y|}.
\label{eqn:upper-inner-product}
\end{align}
Now observe that
$1-2\Pi_{x,y}$ ranges in $[F(x,y)-2\epsilon,F(x,y)+2\epsilon]$ on
$\dom F,$ and 
is bounded in absolute value by $1$ otherwise.
This gives
\begin{align}
\sum_{x\in X}\sum_{y\in Y}
\Psi_{x,y}&(1-2\Pi_{x,y})\nonumber\\&\geq
\sum_{(x,y)\in \dom F} (\Psi_{x,y}F(x,y)-2\epsilon|\Psi_{x,y}|)
-\sum_{(x,y)\notin \dom F} |\Psi_{x,y}|\nonumber\\
&\geq
\sum_{(x,y)\in \dom F}
\Psi_{x,y}F(x,y)-2\epsilon
-\sum_{(x,y)\notin \dom F} |\Psi_{x,y}|,
 \label{eqn:lower-inner-product}
\end{align}
where the last step uses $\|\Psi\|_1\leq1.$
The theorem follows by comparing 
the upper bound (\ref{eqn:upper-inner-product}) with
the lower bound (\ref{eqn:lower-inner-product}).
\end{proof}

We are now in a position to prove
Theorem~\ref{thm:pmm}, which we restate here for the
reader's convenience.
\begin{theorem}[restatement of Theorem~\ref{thm:pmm}]
Let $F$ be the $(n,k,f)$-pattern matrix, 
where $f\colon \zoo^k\to\{-1,+1,*\}$ is given.
Then for every $\epsilon\in[0,1)$ and every $\delta<\epsilon/2,$
\begin{align}
Q^*_{\delta}(F) &\geq
\frac{1}{4}  \degeps(f)\log \left(\frac{n}{k}\right) - 
\frac12 \log\left(\frac{3}{\epsilon-2\delta}\right).
\label{eqn:pattern-matrix-general-error}
\end{align}
\end{theorem}

\begin{proof}
Let $d=\degeps(f)\geq1.$ By Theorem~\ref{thm:dual-approx},
there is a function $\psi\colon \zoo^k\to\Re$ such that:
\begin{align}
&\;\,\hat\psi(S)=0 &&(|S|<d),
	\label{eqn:psi-fourier-coeffs} \\
&\sum_{z\in\zoo^k}|\psi(z)| =1,
	\label{eqn:psi-bounded} \\
&
\sum_{z\in\dom f}f(z)\psi(z)-\sum_{z\notin\dom f}|\psi(z)|>\epsilon.
	\label{eqn:psi-correl}
\end{align}
Let $\Psi$ be the
$(n,k,2^{-n}(n/k)^{-k}\psi)$-pattern matrix. Then 
(\ref{eqn:psi-bounded}) and (\ref{eqn:psi-correl}) show that
\begin{align}
&\|\Psi\|_1 =1, \label{eqn:K-bounded}\\
&\sum_{(x,y)\in\dom F}F_{x,y}\Psi_{x,y}-\sum_{(x,y)\notin\dom F}
|\Psi_{x,y}|>\epsilon.
\label{eqn:K-M-correl}
\end{align}
Our last task is to calculate $\|\Psi\|.$ It follows
from (\ref{eqn:psi-bounded}) that
\begin{equation}
\max_{S\subseteq[k]} |\hat\psi(S)| \leq 2^{-k}.
	\label{eqn:max-fourier-coeff-psi}
\end{equation}
Theorem~\ref{thm:pattern-spectrum} yields, in view of
(\ref{eqn:psi-fourier-coeffs}) and (\ref{eqn:max-fourier-coeff-psi}):
\begin{equation}
\|\Psi\|
\leq  \left(\frac{k}{n}\right)^{d/2} 
      \left(2^{n+k} \left(\frac{n}{k}\right)^k\right)^{-1/2}. 
\label{eqn:K-norm}
\end{equation}
Now (\ref{eqn:pattern-matrix-general-error})
follows from (\ref{eqn:K-bounded}), (\ref{eqn:K-M-correl}), 
(\ref{eqn:K-norm}), and Theorem~\ref{thm:discrepancy-method}.
\end{proof}

\end{document}